\documentclass[conference]{IEEEtran}
\IEEEoverridecommandlockouts

\usepackage{amsmath,amssymb,graphicx,epsfig, cite,algorithm,algorithmic,epstopdf, url}
\usepackage[utf8]{inputenc}
\usepackage[mathscr]{euscript}
\usepackage{color}
\usepackage{bm}
\usepackage{amsthm}

\def\minwrt[#1]{\underset{#1}{\text{minimize }}}
\def\argminwrt[#1]{\underset{#1}{\text{arg min }}}
\def\argmaxwrt[#1]{\underset{#1}{\text{arg max }}}
\def\maxwrt[#1]{\underset{#1}{\text{maximize }}}
\def\maxemphwrt[#1]{\underset{#1}{\text{\emph{maximize} }}}

\newtheorem{remark}{Remark}
\newtheorem{proposition}{Proposition}

\newcommand{\norm}[1]{\left\lVert#1\right\rVert}

\def\RN{{\mathbb{N}}}
\def\RR{{\mathbb{R}}}

\newcommand{\dualvar}{\mathcal{G}}

\newcommand{\dualset}{\chi}

\newcommand{\dualsubset}{\mathcal{K}}

\newcommand{\coeffunc}{\mathcal{\xi}}

\newcommand{\Lagrangefunc}{\mathcal{L}}

\newcommand{\psdcone}{\mathcal{C}}

        \makeatletter
        \def\fps@eqnfloat{!t}
        \def\ftype@eqnfloat{4}
        
        \newenvironment{eqnfloat*}
               {\@dblfloat{eqnfloat}}
               {\end@dblfloat}
        \makeatother
\renewenvironment{thebibliography}[1]{%
 \begin{oldthebibliography}{#1}%
 \setlength{\parskip}{0ex}%
 \setlength{\itemsep}{0ex}%
}%
{%
\end{oldthebibliography}%
}%

\ifCLASSINFOpdf
\else
\fi
\begin{document}
\title{An efficient solver for \\designing optimal sampling schemes}

%
\author{\IEEEauthorblockN{Filip Elvander\IEEEauthorrefmark{1}, Johan Sw\"ard, and
Andreas Jakobsson\IEEEauthorrefmark{2}}
\IEEEauthorblockA{\IEEEauthorrefmark{1}Dept. of Electrical Engineering (ESAT-STADIUS), KU Leuven, Belgium\\
}
\IEEEauthorblockA{\IEEEauthorrefmark{2}Centre for Mathematical Sciences, Lund University, Sweden
}

}


\maketitle


\begin{abstract}
In this short paper, we describe an efficient numerical solver for the optimal sampling problem considered in \textit{Designing Sampling Schemes for Multi-Dimensional Data} \cite{SwardEJ18_150}. An implementation may be found on\\ \textit{https://www.maths.lu.se/staff/andreas-jakobsson/publications/}.
\end{abstract}
\vspace{2mm}
%
%
\section{Problem statement}
For a background to the optimal sampling problem, see \cite{SwardEJ18_150}. Consider the signal model
\begin{align*}
	y(t) \sim p(\cdot; t, \theta)
\end{align*}
where $p$ is a probability density function parametrized by the sampling parameter $t \in \RR^s$ and the parameter vector $\theta \in \RR^P$. Here, $\theta$ is the parameter of interest to be estimated. Assume that we get to choose to sample $y$ at $K$ out of $N$ potential samples $t_n$, $n = 1,\ldots,N$. We then want to solve
\begin{equation} \label{eq:opt_problem}
\begin{aligned}
	\minwrt[w \in \mathcal{W},\mu \in \RR^P]& \quad\sum_{p=1}^P \psi_p\mu_p \\
	\text{subject to }& \begin{bmatrix}
		\sum_{n=1}^N w_n F_n(\theta) & e_p \\
		e_p^T & \mu
	\end{bmatrix} \succeq 0 \\ &p = 1,\ldots,P,
\end{aligned}
\end{equation}
where
\begin{align*}
\mathcal{W} = \left\{ w \in \RR^N \mid \sum_{n=1}^N w_n = K \;,\; w_n \in [0,1] \right\}
\end{align*}
is the set of allowed weights, indicating which $K$ samples that are selected, $\mu \in \RR^P$ correspond to the Cram\'er-Rao lower bound for $\theta$, and where $\psi \in \RR_+^P$ is a vector of non-negative weights. Also, $e_p$ is the $p$th canonical basis vector. Herein, we describe how to solve \eqref{eq:opt_problem} efficiently by considering its dual problem.\footnote{An implementation of the solution algorithm may be found on \textit{https://www.maths.lu.se/staff/andreas-jakobsson/publications/}.}

%
\section{Dual problem}
Consider the Lagrangian relaxation of \eqref{eq:opt_problem} according to
\begin{align*}
	\mathcal{L} = \sum_p \psi_p \mu_p -\sum_{p=1}^P \left\langle G_p , \begin{bmatrix}
		\sum_{n=1}^N w_n F_n(\theta) & e_p \\
		e_p^T & \mu_p
	\end{bmatrix}\right\rangle,
\end{align*}
where $G_p$, $p = 1,\ldots,P$, are dual variables, i.e., positive semi-definite matrices of dimension $P\times P$. Let $G_p$ be structured according to
\begin{align}\label{eq:dual_partioning}
	G_p = \begin{bmatrix}
		\tilde{G}_p& \gamma_p \\
		\gamma_p^{T} & g_p
	\end{bmatrix}.
\end{align}
%
%
\begin{algorithm}[t]\caption{Sub-gradient ascent.}\label{alg:subgrad_ascent}
	\begin{algorithmic}
		\STATE Require: Initial guess $\dualvar = \left\{G_{p}  \right\}_{p=1}^P $, step size $\alpha$.
		\WHILE{Not converged}
		\STATE Find $w \in \argminwrt[w \in \mathcal{W}]-\sum_{n=1}^Nw_n \coeffunc_n(\dualvar)$.
		\FOR{p=1:P}
		\STATE $\mu_p \leftarrow e_p^T\left( \sum_{n=1}^N w_n F_n(\theta) \right)^{-1}e_p$.
		\ENDFOR
		\FOR{p=1:P}
		\STATE $G_p \leftarrow \mathcal{P}_{\mathcal{K}_{\psi_p}}\left(G_{p}  + \alpha \nabla q_p(\dualvar)\right)$.
		\ENDFOR
		\STATE $\dualvar \leftarrow \left\{G_{p}  \right\}_{p=1}^P $.
		\ENDWHILE
		\RETURN $w \in \argminwrt[w \in \mathcal{W}]-\sum_{n=1}^Nw_n \coeffunc_n(\dualvar)$.
	\end{algorithmic}
\end{algorithm}
%
%
For notational convenience, let a dual point be denoted
\begin{align}
	\dualvar = \left\{ G_p \right\}_{p=1}^{P}
\end{align}
and define
\begin{align}
	\coeffunc_n(\dualvar) = \left\langle F_n(\theta), \sum_{p=1}^{P} \tilde{G}_p\right\rangle
\end{align}
Then, for any $w$,
\begin{align*}
	\inf_\mu \Lagrangefunc =\begin{cases}
	-\sum_{n=1}^Nw_n \coeffunc_n(\dualvar) - 2\sum_{p=1}^{P} e_p^T \gamma_p,& \text{ if } g_p= \psi_p, \\
	-\infty & \text{ otherwise}.
	\end{cases}
\end{align*}
The infimum with respect to $w \in \mathcal{W}$ is given by setting the $K$ entries corresponding to the $K$ largest values of $\xi_n(\dualvar)$ equal to 1 and the rest to zero. Note that the minimizing $w$ is not necessarily unique. Specifically, if the $K+1$:th largest value of $\xi_n(\dualvar)$ is strictly smaller than the $K$:th largest, the minimizing $w$ is unique. Otherwise, there are infinitely many solutions.
Thus, the dual problem is
\begin{equation*}
\begin{aligned}
	\maxwrt[\substack{G_p\succeq0 \\p = 1,\ldots,P}]& \inf_{w \in \mathcal{W}} -\sum_{n=1}^Nw_n \coeffunc_n(\dualvar) - \sum_{p=1}^{P} e_p^T \gamma_p\\
	\text{subject to }& \quad g_p = \psi_p \;,\; p = 1,\ldots,P.
\end{aligned}
\end{equation*}
Letting $E = e_P e_P^T$, we may express the constraint as
\begin{align*}
	 \left\langle G_p, E\right\rangle = \psi_p \;,\; p = 1,\ldots,P.
\end{align*}
Thus, defining the family of sets parametrized by $\phi$,
\begin{align} \label{eq:dual_subset}
	\dualsubset_{\phi} = \left\{ U  \mid \left\langle  U, E\right\rangle = \phi \;,\; U\succeq 0 \right\}
\end{align}
and letting
\begin{align*}
	\dualset_\psi = \left\{ \dualvar \mid G_p \in \dualsubset_{\psi_p}\;,\; p= 1,\ldots,P\right\}
\end{align*}
we may express the dual problem as
\begin{equation} \label{eq:dual_problem}
\begin{aligned}
	\maxwrt[\dualvar\in \dualset_\psi] q(\dualvar)
\end{aligned}
\end{equation}
where the dual objective function is
\begin{align} \label{eq:dual_func}
	q(\dualvar) = \inf_{w \in \mathcal{W}} -\sum_{n=1}^Nw_n \coeffunc_n(\dualvar) - \sum_{p=1}^{P} e_p^T \gamma_p.
\end{align}
We utilize the ideas from Nedic and Ozdaglar \cite{NedicO09_19} in order to maximize the dual problem \eqref{eq:dual_problem} using sub-gradient ascent. The algorithm is summarized in Algorithm~\ref{alg:subgrad_ascent}. A short derivation of the step is presented in the following sections.
\subsection{Sub-gradient ascent}
For a dual point $\dualvar \in \dualset_\psi$, a sub-gradient of $q$ in \eqref{eq:dual_func} at $\dualvar$, denoted $\nabla q(\dualvar)$, can be decomposed in components $\nabla q_p(\dualvar)$, where each component is given by
\begin{align}
	\nabla q_p(\dualvar) = - \begin{bmatrix}
		\sum_{n=1}^N w_n F_n(\theta) & e_p \\ e_P^T & \mu_p
	\end{bmatrix}
\end{align}
where
\begin{align}
	(w,\mu_p) \in \argminwrt[w,\mu_p] \Lagrangefunc(\mu,w,\dualvar).
\end{align}
As noted earlier, one may retrieve a primal vector $w$ from this set setting the entries of $w$ corresponding to the $K$ largest values of $\left\{ \coeffunc_n(\dualvar) \right\}_{n=1}^N$ to 1 and the rest to zero. Noting that any $\mu_p \in \RR$ is a member of the minimizing set, one may here choose
\begin{align}
	\mu_p = e_p^T\left( \sum_{n=1}^N w_n F_n(\theta) \right)^{-1}e_p,
\end{align}
i.e., the $\mu_p$ minimizing the primal objective, while still retaining primal feasibility for this choice of $w$. Then, a dual ascent method guaranteeing that the dual variable $\dualvar$ is feasible may be realized according to
\begin{align*}
	G_p \leftarrow \mathcal{P}_{\mathcal{K}_{\psi_p}}\left(G_{p}  + \alpha \nabla q_p(\dualvar)\right)
\end{align*}
for $p = 1,\ldots,P$, where $\mathcal{P}_{\mathcal{K}_{\phi}}$ denotes projection on $\mathcal{K}_{\phi}$, as defined in \eqref{eq:dual_subset}. How to perform this projection is described in the next section.
%
\subsection{Projection on PSD cone with constraint}
Consider a set of $\mathcal{G} = \left\{G_m\right\}_{m=1}^M$ of $M \in \RN$ symmetric matrices $G_m \in \RR^{P\times P}$. Let $\psdcone$ be the set of $P\times P$ positive semidefinite matrices. Here, we are interested in computing the projection on the set
\begin{align}
	\mathcal{K}_{\phi} = \left\{ U \mid \left\langle U, E\right\rangle = \phi \;,\; U \in \psdcone \right\}
\end{align}
for $\phi \in \RR_+$ and a symmetric matrix $E \in \psdcone$.
\begin{proposition}\label{prop:projection}
	The projection on $\mathcal{K}_{\phi}$, denoted $\mathcal{P}_{\mathcal{K}_{\phi}}$ is given as
	\begin{align}
		\mathcal{P}_{\mathcal{K}_{\phi}}: G \mapsto \mathcal{P}_\psdcone(G + \lambda E)
	\end{align}
	where $\mathcal{P}_\psdcone$ denotes projection on $\psdcone$, and where $\lambda \in \RR$ is the unique root of the equation
\begin{align}
	\left\langle \mathcal{P}_\psdcone(G + \lambda E), E\right\rangle = \phi.
\end{align}
\end{proposition}
%
%
\begin{proof}
See appendix.
\end{proof}
%
\begin{remark}
It may here be noted that projecting on $\psdcone$ is simply performed by computing an eigenvalue decompostion and setting all negative eigenvalues to zero.
\end{remark}
\subsection{Computational complexity}
It may be noted that finding $(w,\mu) \in \argminwrt[w,\mu] \Lagrangefunc(\mu,w,\dualvar)$ is linear in $N$ and quadratic in $P$. Performing the gradient step is linear in $P$, whereas the projection on the dual feasible set is $\mathcal{O}(P^3)$. To see this, it may be noted that in practice, one may solve the equation $\left\langle \mathcal{P}(G + \lambda E), E\right\rangle = \phi$ using interval halving techniques, where each evaluation of the right-hand side requires computing one eigenvalue decomposition. The per-iteration complexity for this scheme is thus $\mathcal{O}(P^3)$.
\newpage
\appendix
%
\begin{proof}[Proof of Proposition~\ref{prop:projection}]
By definition, $U= \mathcal{P}_{\mathcal{K}_{\phi}}(G)$ solves
\begin{align}
	\minwrt[U \in \mathcal{K}_{\phi}] \frac{1}{2}\norm{U - G}_F^2,
\end{align}
where $\norm{\cdot}_F$ is the Frobenius norm.
To arrive at a dual formulation, consider the Lagrangian
\begin{align*}
	\tilde{\mathcal{L}} = \frac{1}{2} \norm{U - G}_F^2 - \lambda \left( \left\langle U, E\right\rangle - \phi \right) - \langle \Lambda, U\rangle,
\end{align*}
with dual variables $\Lambda \in \psdcone$ and $\lambda \in \RR$.
We may complete the square according to
\begin{align*}
	&\frac{1}{2}\norm{U - G}_F^2 -  \langle \lambda E  +\Lambda, U\rangle \\
	&= \frac{1}{2}\norm{U - \left( G + \lambda E  +\Lambda \right)}_F^2 -  \frac{1}{2}\norm{ G + \lambda E  +\Lambda}_F^2 - \frac{1}{2}\norm{ G}_F^2.
\end{align*}
Then, the infimum of $\tilde{\mathcal{L}}$ with respect to $U$ is given by
\begin{align*}
	\inf_{U}\tilde{\mathcal{L}} = - \frac{1}{2}  \norm{ G + \lambda E  +\Lambda}_F^2 + \norm{G}_F^2 + \phi \lambda,
\end{align*}
which is attained for $U = G + \lambda E  +\Lambda$. Consider the dual function
\begin{align}
	r(\Lambda,\lambda) = - \frac{1}{2}  \norm{ G + \lambda E  +\Lambda}_F^2 + \phi \lambda.
\end{align}
For each $\lambda$, this is maximized with respect to $\Lambda \in \psdcone$ by
\begin{align}
	\Lambda = \mathcal{P}_\psdcone\left( -(G + \lambda E)  \right),
\end{align}
i.e., $\Lambda$ is constructed from the negative part of the eigendecomposition of $G + \lambda E$, with flipped sign. Using
\begin{align}
	G + \lambda E = \mathcal{P}_\psdcone\left( G + \lambda E \right)+ \mathcal{P}_\psdcone\left( -(G + \lambda E)  \right),
\end{align}
this yields
\begin{align}
	\tilde{r}(\lambda) = \sup_{\Lambda} r = -\frac{1}{2} \norm{\mathcal{P}_\psdcone\left( G + \lambda E \right)}_F^2 + \phi \lambda.
\end{align}
This one-dimensional criterion may then be maximized with respect to $\lambda \in \RR$. However, as we from the analysis obtain $U = \mathcal{P}_\psdcone\left( G + \lambda E \right)$, we may utilize the primal feasibility condition $\left\langle  U, E\right\rangle = \phi$. Specifically, one may seek the roots of
\begin{align}
	f(\lambda) = \left\langle  \mathcal{P}_\psdcone(G + \lambda E), E\right\rangle- \phi.
\end{align}	
As $E \in \psdcone$, $f$ is a continuous, monotone increasing function, and $f$ thus has a unique zero. 
\end{proof}
\bibliographystyle{IEEEbib}
\bibliography{sampling_solver_arxiv.bbl}
\end{document}